\documentclass{article}
\usepackage[margin=1.2in]{geometry}
\usepackage{graphicx} 
\usepackage{amsmath}
\usepackage{amssymb}
\usepackage{xcolor}
\usepackage{graphicx}
\usepackage{wrapfig}   
\usepackage{amsfonts, amsbsy, amsthm, latexsym, epsfig}
\newcommand{\R}{\mathbb{R}}
\newcommand{\Z}{\mathbb{Z}}

\newcommand{\CC}{\mathbb{C}}

\newcommand{\mv}{m}
\newcommand{\SB}{\mathcal{B}}
\newcommand{\mvb}{\mathfrak m}
\newcommand{\PPW}{\Pi_k}
\usepackage{algorithm}
\usepackage{booktabs}
\usepackage{algorithmicx}
\usepackage[noend]{algpseudocode}
\usepackage{tikz}
\usepackage{subfig}
\usepackage{bbm}
\usepackage{comment}
\usepackage{tabularx}

\newtheorem{definition}{Definition}
\newtheorem{theorem}{Theorem}
\newtheorem{corollary}{Corollary}
\newtheorem{proposition}{Proposition}
\newtheorem{remark}{Remark}
\definecolor{akcolor}{rgb}{0.65, 0.15, 0.6}
\newcommand{\ak}[1]{\textcolor{akcolor}{#1}}

\definecolor{ikcolor}{rgb}{0.05, 0.15, 0.6}
\newcommand{\ik}[1]{\textcolor{ikcolor}{#1}}

\definecolor{vbcolor}{rgb}{0.35, 0.15, 0.6}
\newcommand{\vb}[1]{\textcolor{vbcolor}{#1}}

\definecolor{bmcolor}{rgb}{0.9, 0.3, 0}
\newcommand{\bm}[1]{\textcolor{bmcolor}{#1}}

\DeclareMathOperator*{\argmax}{arg\,max}
\DeclareMathOperator*{\argmin}{arg\,min}

\title{Reconstructing Graph Signals from Noisy Dynamical Samples}
\author{Akram Aldroubi, Victor Bailey, Ilya Krishtal, Brendan Miller, Armenak Petrosyan}

\begin{document}

\maketitle

\begin{abstract}

We investigate the dynamical sampling space-time trade-off problem within a graph setting. Specifically, we derive necessary and sufficient conditions for space-time sampling that enable the reconstruction of an initial band-limited signal on a graph. Additionally, we develop and test numerical algorithms for approximating the optimal placement of sensors on the graph to minimize the mean squared error when recovering signals from time-space measurements corrupted by i.i.d.~additive noise. Our numerical experiments demonstrate that our approach outperforms previously proposed algorithms for related problems.

\end{abstract}

\section{Introduction}

Sampling and reconstruction is a fundamental problem in the study of various signals, for example in communication 
\cite{AG01, Unser:00}.   In modern applications,  there are many situations  where the functions to be sampled and recovered are evolving in time.  For example, in video processing  
a function $f$ that models the video depends not just on space but also on time. In other instances, time dependence may describe diffusion, transport, or other physical phenomena \cite{ GRUV15,JTSS20, LMT21, LDV11, STWZ20}. The classical sampling framework, however,  is not well suited for such situations, because it does not differentiate between spatial and temporal variables. The time direction plays a special role that differs from the spatial dimensions among which there are, often, no preferential directions. 
%
In some situations, however, spatial directions are limited and form a graph-like structure. Thus, the sampling and reconstruction of functions evolving on graphs are crucial in modeling signals
across various networks like cellular and internet communications. These signals often represent data related to epidemiology, environmental science, and other areas \cite{BISVP08, RDCV12}. 

\smallskip \noindent When studying a 
signal that
varies in time 
because of an underlying evolution process, designing a space-time sampling procedure must take into account the driving evolution operator (see, e.g., \cite{LV09, MBD15, RCLV11, RDCV12}).  The mathematical theory for the reconstruction of space-time signals is recent  and is commonly called {\em dynamical sampling} \cite { AKh17,  AVVP23, ACCMP17,  ADK13, AP23, B23, BK23, CMPP20, CH19, HX22, T10, UZ21, ZLL17_2}. A related active area of research is called {\em mobile sampling} \cite{RBD21, GRUV15, JM22, JMV24}.

\smallskip \noindent  Several problems in dynamical sampling where $\{(t_i,x_j)\}\subset [0,\infty) \times \R^d$ or $\{(t_i,x_j)\}\subset [0,\infty) \times \Z^d$ have recently been solved \cite { AGHJKR21, UZ21}.

\smallskip \noindent  In this paper, we study the space-time sampling trade-off for finite graphs. We follow a deterministic approach as opposed to the probabilistic one in \cite{HNT21}. Specifically,
for a graph $G=(V,E)$, a graph operator $A$, 
and Paley-Wiener space $PW_\omega$ (see Definitions \ref {IVPG} and \ref {PW} below), we consider an initial value problem  (IVP) 
given by 
\begin{equation}\label{DFM1}
	\begin{cases}
	\dot{u}(t)=Au(t)\\
	u(0)=u_0,
	\end{cases}
	\quad t\in\mathbb R_+,\ u_0\in PW_\omega,
\end{equation} 
where $u$ is a vector-valued function defined on the nodes $V = \{1,\ldots, d\}$ of $G$. We note that, due to the fact that $A$ is a graph operator, the condition  $u_0\in PW_\omega$ implies that the values of $u$ remain in $PW_\omega$, i.e.,  $u:\R_+\to PW_\omega\subset \ell^2(V)$. We investigate the problem of sensor placement for the recovery of $u_0$ from the corresponding space-time samples of $u$.

\subsection{Contributions.}
\smallskip \noindent
The contributions of this paper are as follows.

\smallskip \noindent
First, in Section \ref{RPW}, we provide necessary and sufficient conditions on a set of vertices where the solution $u$ needs to be sampled in time to recover any initial condition $u_0 \in PW_\omega$. 

\smallskip \noindent
Next, in Section \ref{ONS}, given a set of vertices $\Omega \subset V$ for which recovery is possible, 
we look for an optimal such set with respect to recovery from noisy space-time samples. In particular, we consider 
samples 
with i.i.d.~additive noise and wish to minimize the mean squared error (MSE) of recovery using a sampling set of vertices of fixed size $|\Omega| = s$. 

\smallskip \noindent
Finding an exact solution to the above problem is usually prohibitively computationally expensive as the problem is combinatorial in nature. Thus, in Section \ref{NA}, we explore numerical algorithms for approximating the solution. In particular, we present a greedy algorithm and a few algorithms based on various convex relaxations of the original problem. 

\begin{figure}[H]
    \centering
    {\includegraphics[width=0.45\linewidth]{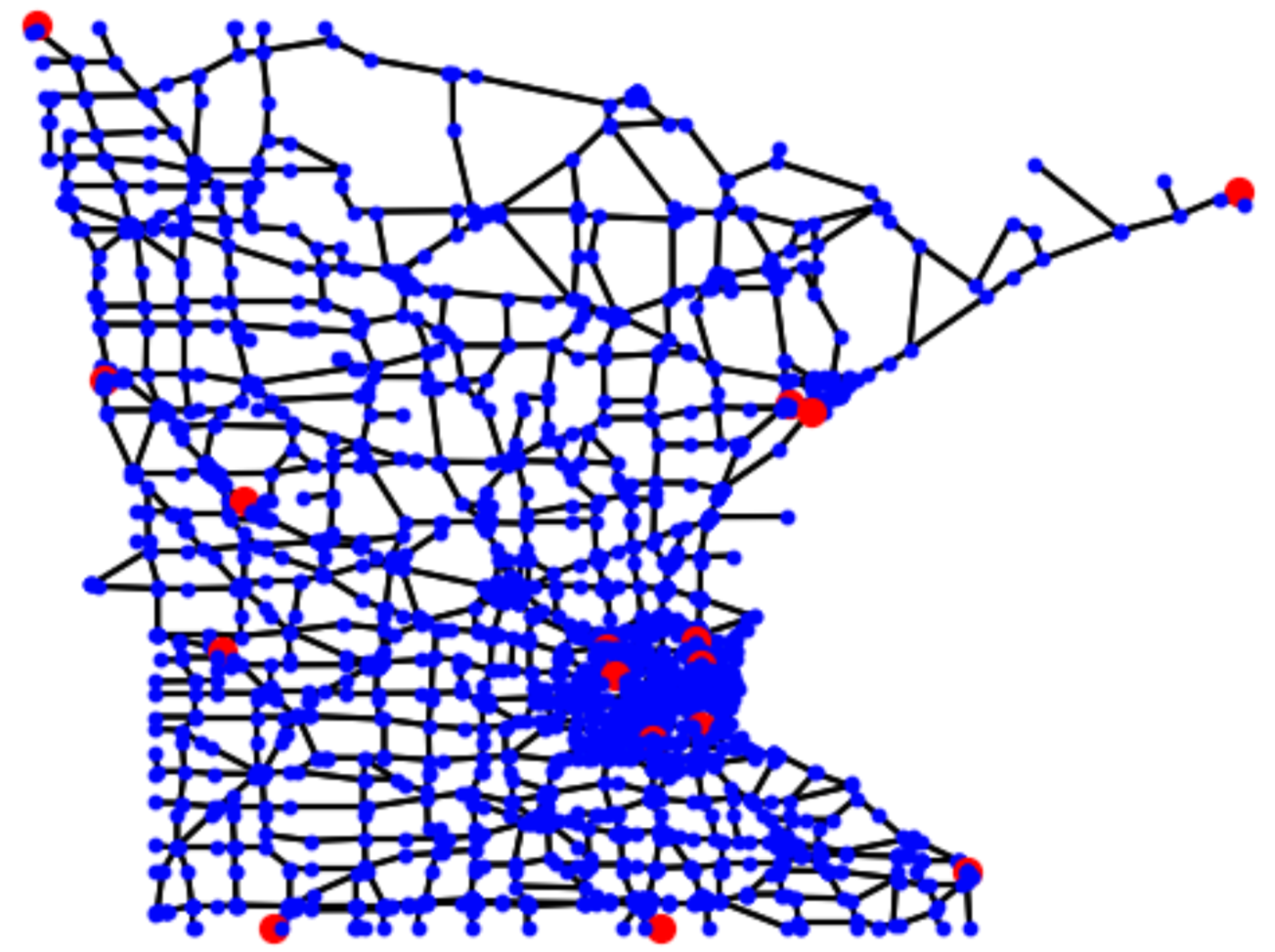}}
    \qquad
    {\includegraphics[width=0.45\linewidth]{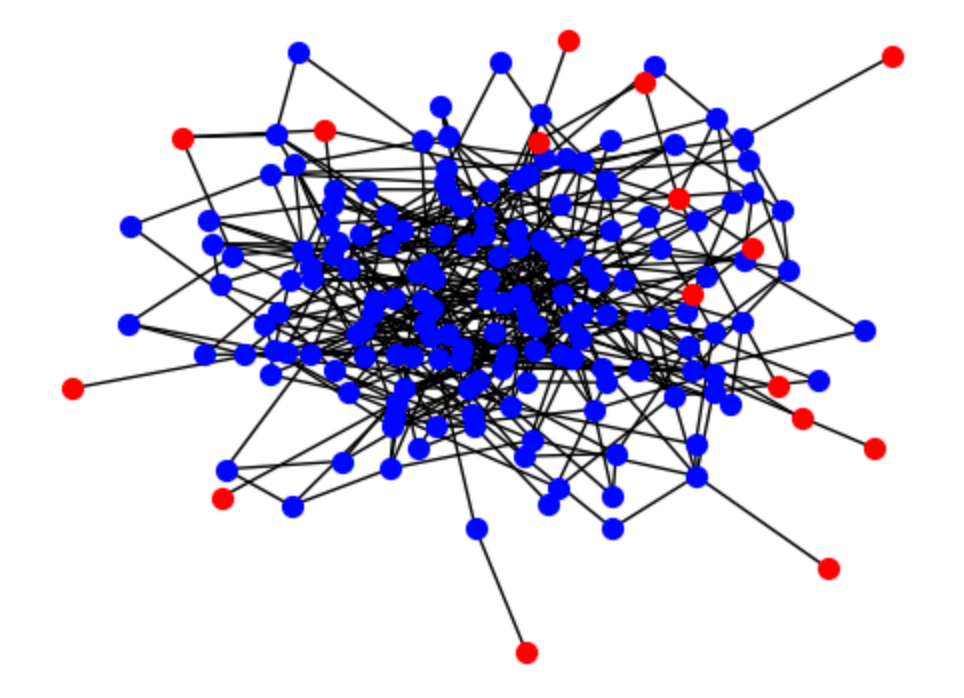}}
    \caption{Greedy algorithm selecting 16 sampling locations (in red) to minimize the MSE of signal reconstruction over $PW_{\lambda_{10}}$ on the Minnesota roadmap graph (left) and a random connected graph (right), where $\lambda_{10}$ is the 10-th smallest eigenvalue of the graph operator.}
    \label{Graphs}
\end{figure}

\smallskip \noindent
Finally, we compare the relative numerical performance of our algorithms in Section \ref{NP}. 
 We note that our tests results lead us to a heuristic conclusion that the optimal placement of sensors tends to favor ``remote'' vertices of the graph, i.e.~those that are not easily accessible from the rest of the graph (see Fig.~\ref{Graphs} for an illustration).


\subsection{Basic definitions and notation.} 
 \noindent Symbols $G=(V,E)$  will denote a finite undirected graph, where $V = V(G)= \{1,\ldots,d\}$  is the set of vertices, $d = |V|$ is the cardinality of $V$,  and $E = E(G) \subseteq V\times V$ is the set of edges.

 \smallskip
 
\noindent The degree of  vertex $i$, denoted by $deg(i)$, is the number of edges that are connected to $i$.
\smallskip

\noindent 
Next, we define a few standard operators associated with a graph. In this paper, we routinely identify operators on $\mathbb C^d$ with their matrices in the standard basis.

 \smallskip
\noindent For a graph $G = (V, E)$, $V = V(G)= \{1,\ldots,d\}$, the Laplacian $\Delta$ of $G$ is the $d \times d$ matrix $\Delta = D - C$, where $D$ is a diagonal matrix such that $d_{ii} = deg(i)$ and  $C$ is the $d \times d$ adjacency matrix of $G$, i.e.,  $c_{ij} = 1$ if there is an edge between $i$ and $j$ and $c_{ij} = 0$ otherwise. Clearly, 
the Laplacian $\Delta$ is a symmetric matrix. Therefore, the eigenvalues of $\Delta$ are real and $\Delta$ has an orthonormal set of eigenvectors. 
\smallskip


\smallskip

\noindent In order for the IVP \eqref{DFM1} to be associated with a graph $G$ we must make restrictions on the possible operators $A$. To this end, we use the Laplacian operator $\Delta$ on $G$ to relate $A$ to the topology of $G$ similarly to the way it was done in \cite {IBLL20}.

\begin{definition} \label {IVPG}
An IVP \eqref{DFM1} is associated with a graph $G$ if the operator $A$ is positive semidefinite and has eigenvectors coinciding with those of the Laplacian $\Delta$. Such an operator $A$ will be called a \emph{graph operator} (associated with $G$).
\end{definition}
\smallskip

\noindent 
In particular, a graph operator commutes with the Laplacian of $G$.
Decomposing $\Delta$ as $B^*D_\Delta B$ where $B$ is unitary and $D_\Delta$ is diagonal, we get that $A$ is a graph operator if and only if $A$ can be written as $A=B^*D_AB$, where $D_A$ is diagonal with real non-negative entries. Thus, a graph operator $A$ is orthogonally diagonalizable. 
\smallskip

\noindent
In the following definition, we introduce Paley-Wiener spaces on graphs as it was done in \cite{Pes08}.

\begin{definition} \label{PW}
The \emph{Paley-Wiener space} $PW_{\omega}(A) \subseteq \ell^2(V) \simeq  \mathbb{C}^d$ of  a graph operator $A$ on $G$ is the span of the eigenvectors corresponding to its eigenvalues 
bounded above by $\omega$. In particular, if $A$ has $n \leq d$ distinct real eigenvalues $0\le\lambda_1 < \lambda_2 < \ldots <\lambda_n$, then the space $PW_{\lambda_k}(A)$,  where $k \leq n$, is the span of the eigenvectors of $A$ corresponding to the eigenvalues $\{\lambda_1, \ldots, \lambda_k\}$. 
\end{definition}

\noindent If there is no ambiguity regarding the graph operator, then we will write $PW_{\omega}$ instead of $PW_{\omega}(A)$. Additionally, in the case where $\omega = \lambda_k$, for some eigenvalue $\lambda_k$ of $A$, we will denote the Paley-Wiener space $PW_{\lambda_k}$ as $PW_{k}$.

\smallskip

\section {Reconstruction in Paley-Wiener Spaces} \label{RPW}
The main result in this section is Theorem \ref{mainthm} which provides necessary and sufficient conditions on which vertices need to be sampled in time to recover any initial condition $u_0$. This result is based on Theorem \ref{TAD0} from \cite{ACMT17}, provided below. Before we begin, we introduce the following definitions and notation.

\smallskip
\noindent Let $D=diag(\lambda_1 I_1, \ldots,\lambda_nI_n)$ be  a $d\times d$ diagonal matrix where $I_i$ are identity matrices of various sizes. We denote by $P_j$ the orthogonal projection in $\CC^d$ onto the eigenspace of $D$ associated to the eigenvalue $\lambda_j$. In particular, $P_j$ is a projection onto a span of a subset of the standard basis $\{e_i:i =1,\ldots,d\}$. 

\smallskip
 \noindent For a $d\times d$ matrix $A$ and a vector $b \in \CC^d$, the $A\text{-annihilator}$ of a vector $b$ is the monic polynomial $q_b(x)=a_0+a_1x+\cdots+x^{r_b}$ of smallest degree such that $q_b(A)b = 0$.


\begin {theorem}[\cite {ACMT17}]\label {TAD0}
Let $\Omega \subset \{1, \dots, d\}$ and $\{b_i: i \in \Omega \}$  vectors in $\CC^d$. Let $D$ be a diagonal matrix and $r_i$ the degree of the $D$-annihilator of $b_i$. Set $l_i=r_i-1$. Then  $\{D^{j}b_i: \; i\in \Omega,  \, j=0, \dots, l_i\}$ is a frame of $\CC^d$ if and only if $\{P_j(b_i):i \in \Omega\}$ form a frame of $P_j(\CC^d)$, $j=1,\dots ,n$.
\end{theorem}
\medskip

\noindent Before we introduce our result on the reconstruction of a functions in a Paley-Wiener space from space-time samples, we need the following definitions and notation.
\begin{definition}
Let $\Omega \subset V =\{1, \ldots, d\}$ be a subset of vertices of a graph $G$. The sub-sampling matrix $S_{\Omega} \in \CC^{d \times d}$ is a diagonal matrix defined by $$(S_{\Omega})_{jj} = \begin{cases} 
      1, & j \in \Omega; \\
      0, & \mbox{otherwise.}
   \end{cases}
$$
\end{definition}

\noindent Let $A \in \CC^{d \times d}$ be a graph operator. Let $W = PW_{k}$ for some $k$ and  let $I_W: W \hookrightarrow \ell^2(V)$ be the natural inclusion. The symbols $\mathcal{A}_{(W,\Omega)}$ will denote the  $Ld \times d$ matrix
 \begin {equation} \label {AOmatrix}
 \mathcal{A}_{(W,\Omega)} = [I_W^*S_\Omega, I_W^*AS_\Omega ,...,I_W^*A^{L-1}S_\Omega]^*.
 \end{equation} 
 
\noindent  When $W = \CC^{d}$, we will denote $\mathcal{A}_{(W,\Omega)}$ as $\mathcal{A}_{\Omega}$. 

 \medskip

\noindent For a given graph operator $A$ with eigenvalues $\lambda_1 < \lambda_2 < \ldots <\lambda_n$, we denote by $Q_j$ the orthogonal projection in $\CC^d$ onto the span of the eigenvectors of $A$ corresponding to the eigenvalue $\lambda_j$. Given a Paley-Wiener space $PW_k$ we denote by $\PPW$ the orthogonal projection onto $PW_k$.

\medskip
\noindent Let $A$ be a graph operator. Then $A=B^*DB$ and every $f \in PW_{k}$ can be written as $f=B^*_1x$ (where $B_1$ is the submatrix of $B$ such that the columns of $B^*_1$ are the eigenvectors associated with $PW_{k}$, and $x \in \CC^p$ with $p=\dim PW_{k}$). Define $b_{1i}$ as the $i$-th column of $B_1$. 
Let $r_{1i}$  be the degree of the $D_k$ annihilator of $b_{1i}$, where $D_k=diag(\lambda_1 I_1, \dots,\lambda_kI_k)$, and $\ell_{1i}=r_{1i}-1$. Using this notation we have the following result. 

\begin{theorem} \label {mainthm} Let $A$ be a graph operator on $G= (V, E)$, $\Omega \subset V = \{1, \ldots d\}$, and consider the Paley-Wiener space $PW_{k}$. Then the following are equivalent.
\begin {enumerate}
\item Any $f \in PW_{k}$ can be recovered from the space-time samples $\{f(i), Af(i), \dots, A^{l_{1i}}f(i): i \in \Omega\}$. 
\item The set of vectors $\{P_j(b_{1i}):i \in \Omega\}$ form a frame for the subspaces $P_j(\CC^d)$, $j=1,\dots ,k$.
\item The set $\{ \PPW A^{l}e_i : l \in \{0, 1, \ldots r_{1i}-1 \}$, $i \in \Omega\}$, where $r_{1i}$ is the degree of the $A$-annihilator of $e_i$, is a frame for $PW_{k}$. That is, $\{ A^{l}e_i : l \in \{0, 1, \ldots r_{1i}-1 \}$, $i \in \Omega\}$ is an outer frame for $PW_{k}$.
\item  The set of vectors $\{Q_j(e_{i}) : i \in \Omega\}$ form a frame for the subspaces $Q_j(\mathbb{C}^d)$, $j=1,\dots ,k$. 
\item The matrix  $\mathcal{A}_{(W,\Omega)}$ in \eqref {AOmatrix} has a left inverse.
\end {enumerate}
\end{theorem}
\begin{proof}

\noindent We start with proving that (1) and (2) are equivalent. Let $D_k=diag(\lambda_1 I_1, \dots,\lambda_k I_k)$.   Since $ f \in PW_{k}$, we can write $f=B_1^*x$ where $x \in \CC^p$ and $p=\dim PW_{k}$. Then, 
\[
\begin{split}A^mf(i)&=\langle A^mf,e_i\rangle\\
&=\langle D_k^mx,B_1e_i\rangle=\langle x,D_k^mb_{1i}\rangle.
\end{split}
\]
Thus, if  $\{P_j(b_{1i}):i \in \Omega\}$ form a frame of $P_j(\CC^p)$, $j=1,\dots ,k$, then using Theorem \ref {TAD0} and the identity above, $x \in \CC^p$ can be recovered from $\{f(i), Af(i), \dots, A^{l_{1i}}f(i): i \in \Omega\}$. Hence, $f=B_1^*x$ can be recovered from the space-time samples $\{f(i), Af(i), \dots, A^{l_{1i}}f(i): i \in \Omega\}$. 

\smallskip \noindent For the converse, assume that for some $j_0 \in \{1,\dots ,k\}$, $\{P_{j_0}(b_{1i}):i \in \Omega\}$ is not a frame for $P_{j_0}(\CC^p)$. Then, there exists a vector $x_0\ne \bf{0}$ such that  $x_0\perp span\{b_{1i}: i \in \Omega\}$ and for  $f_0=B_1^*x_0 \in PW_{k}$ we have $A^mf_0(i)=\langle A^mf_0,e_i\rangle = \langle x_0,D_k^mb_{1i}\rangle=0$ for each $i\in \Omega$, yet $f_0 \ne \bf{0}$.

\noindent We now prove (2) and (3) are equivalent. Let $f \in PW_k.$ Then $f=B_1^*x$ (where the columns of $B_1^*$ form an orthonormal set and $x \in \CC^p$). Observe that
\[
\langle x,D_k^mb_{1i}\rangle=\langle A^mf,e_i\rangle\\
=\langle  f,  \PPW A^me_i\rangle.
\]
Since $\|f\|^2 = \|\sum_{j = 1}^p x_jb^*_{1j}\|^2 = \sum_{j = 1}^p|x_j|^2 = $$ \|x\|^2$ and as it follows from the observation above that 

\[
\sum\limits_{i=1}^{r_{1i}-1} |\langle x,D_k^mb_{1i}\rangle|^2=\sum\limits_{i=1}^{r_{1i}-1} |\langle f,\PPW A^{m}e_i\rangle|^2
\]
we have  $\{ \PPW A^{m}e_i : m \in \{0, 1, \ldots r_{1i}-1 \}$, $i \in \Omega\}$ is a frame for $PW_k$ if and only if $\{ D_k^mb_{1i} : m \in \{0, 1, \ldots r_{1i}-1 \}$, $i \in \Omega\}$ is a frame for $\CC^p$. 
Applying Theorem \ref{TAD0}, the equivalence of (2) and (3) follows.




\smallskip \noindent We now show that (2) and (4) are equivalent. Let $V_j=P_j(\CC^d)=span \{ e_{1},\dots,e_{m}\}$ and $W_j=Q_j(\CC^d)=span \{ B^\star e_{1},\dots,B^\star e_{m}\}$, where $m= \dim V_j= \dim W_j$. 
Clearly, $W_j =B^\star V_j$. Assume that $\{P_jb_{i}: i \in \Omega\}$ is a frame for $V_j$. Then, for any $v\in V_j$, using the fact that $P_jv=v$, we have that 
\begin{displaymath}
\alpha \|v\|^2\le \sum\limits_{i \in \Omega} |\langle v, P_jBe_{i}\rangle|^2= \sum\limits_{i \in \Omega} |\langle v, Be_{i} \rangle|^2\le \beta \|v\|^2.
\end{displaymath}

\smallskip\noindent For $w \in W_j$, let $w=B^\star v$ for some $v \in V_j$.  Using the fact that $\|w\|=\|B^\star v\|=\|v\|$ and $Q_jw=w$, we get
\[
\begin{split}
\sum\limits_{i \in \Omega} |\langle w, Q_je_i\rangle|^2&= \sum\limits_{i \in \Omega} |\langle w, e_i\rangle|^2\\
&=\sum\limits_{i \in \Omega} |\langle B^\star v, e_i\rangle|^2=\sum\limits_{i \in \Omega} |\langle v, Be_i\rangle|^2\ge \alpha \|v\|^2= \alpha \|w\|^2.
\end{split}
\]
Similarly, we get an upper bound estimate 
\[
\sum\limits_{i \in \Omega} |\langle w, Q_je_i\rangle|^2\le \beta \|w\|^2. 
\]
This proves that (2) implies (4). Starting from the assumptions that $\{Q_je_i: i \in \Omega\}$ is a frame for $W_j$, the exact same steps as before prove that (4) implies (2).

\smallskip \noindent Recovery of $f \in PW_k$ from the space-time samples $\{f(i), Af(i), \dots, A^{l_{1i}}f(i): i \in \Omega\}$ is equivalent to the existence of a left inverse for $\mathcal{A}_{(W,\Omega)}$ since we can recover $f$ from the vector $\mathcal{A}_{(W,\Omega)} f$ if and only if there is a left inverse for $\mathcal{A}_{(W,\Omega)}$. Thus, the equivalence of (5) and (1) follows.
\end{proof}

\section{Optimal node sampling for space-time problem on graphs} \label {ONS}

\subsection{Problem setup}
Let $G$ be a finite graph, and let the time measurements be taken at discrete uniform time intervals. 
In this case, without loss of generality, assume that $t_n=n$. Then, for each fixed time $n \leq L$, the vector  $u_n:=(u(n,v))_{v \in V(G)}$ is given by 
\begin{equation} \label {DIVP}
 u_n=A^nu_0,
 \end{equation}
 where $A$ is a graph operator as defined above in Definition  \ref{IVPG}. 
\smallskip

\noindent
We can determine the mean squared error of reconstructing a vector $u_0 \in \CC^d$ via noisy measurements obtained from space-time samples by applying Proposition 4.1 in \cite{AHKLLV18}. In particular, letting $\eta_n$ be random vectors whose entries are i.i.d random variables of mean zero and standard deviation $\sigma$, the noisy space-time samples are recorded as: 
 \[
 \ y_\Omega = \mathcal{A}_\Omega u_0 + \eta
 \]
 where $\mathcal{A}_\Omega$ is as in \eqref{AOmatrix}, when $W=\CC^d$
 and $\eta = [\eta_0^T,...,\eta_L^T]^T$.
Denoting the error in the least squares solution as $\epsilon_\Omega$, we have that $\mathbb{E}(\|\epsilon_\Omega\|_2^{2}) = \sigma^2 tr(C_\Omega^{-1}) =\sigma^{2} \sum_{j=1}^d 1/ \mu_j (\Omega)$ where $\mu_j (\Omega)$ are the eigenvalues of the matrix $C_\Omega = \mathcal{A}_\Omega^* \mathcal{A}_\Omega$.

\medskip
\noindent
It is therefore natural to consider the problem of minimizing the mean squared error $\text{MSE}(\Omega) := \mathbb{E}(\|\epsilon_\Omega\|_2^{2})$ in reconstructing $u_0$ from $ y_\Omega$ across sets $\Omega$ of sampling locations with a given size $|\Omega|=s$. This is captured by the following optimization problem:
 \begin{equation}\label{trinv_problem}
\begin{aligned}
\min_\Omega &\quad \text{MSE}(\Omega) = \sigma^2 tr(C_\Omega^{-1}) \\
\text{s.t.} &\quad |\Omega| = s.
\end{aligned}
\end{equation}
\begin{remark} In this remark, we collect a few observations with regard to the above problem.
    \begin{enumerate}
        \item 
        The MSE$(\Omega)$ decreases as $L$, the number of time steps, increases, as shown in $\cite{AHKLLV18}$. Similarly, for any $v \in V = V(G)$, $\text{MSE}(\Omega \cup \{v\}) \leq  \text{MSE}(\Omega)$ .
        
        \item For the case in which $u_0 \in W = PW_{k}$ for some $k$ we may similarly construct $\mathcal{A}_{(W, \Omega)}$ as in \eqref{AOmatrix}
 noting that $\mathcal{A}_{(W, \Omega)} v = \mathcal{A}_\Omega v$ for all $v \in W$. 

 \item  When $\mathcal{A}_{(W,\Omega)}u_0$ has full rank (see Corollary 1), $u_0$ can be recovered from $\mathcal{A}_{(W,\Omega)}u_0$ as $u_0 = \mathcal{A}_{(W,\Omega)}^\dagger\mathcal{A}_{(W,\Omega)}u_0$. Let $\eta$ be as defined previously. We assume that the space-time measurements are perturbed by noise  $\eta$ as follows:
 \[
 y_{(W,\Omega)} = \mathcal{A}_{(W,\Omega)} u_0 + \eta.
 \]
 The error of reconstruction is now $\epsilon_{(W,\Omega)} = \mathcal{A}_{(\Omega,W)}^\dagger\eta$. In the following two propositions, we observe that $\mathbb{E}(\|\epsilon_{(W,\Omega)}\|^2) = \sigma^2\text{tr}(\mathcal{A}^*_{(W,\Omega)}\mathcal{A}_{(W,\Omega)})^{-1}$ and $\mathbb{E}(\|\epsilon_{(W,\Omega)}\|^2) \leq \mathbb{E}(\|\epsilon_{\Omega}\|^2) $. 
    \end{enumerate}
\end{remark}



 \begin{proposition}
     Let $W = PW_{\lambda_k}$ and $C_{(W,\Omega)} = \mathcal{A}^*_{(W,\Omega)}\mathcal{A}_{(W,\Omega)}$. With $\epsilon_{(W,\Omega)}$ defined as above, we have $\mathbb{E}(\|\epsilon_{(W,\Omega)}\|^2) = \sigma^2 \text{tr}(C_{(W,\Omega)}^{-1})$.
 \end{proposition}

\smallskip\noindent
 The proof follows immediately from the proof of Proposition 2.5 in \cite{GVT98}, which shows that $\mathbb{E}(\|\Psi^\dagger\eta\|^2) =  \sigma^2\text{tr}(\Psi^*\Psi)^{-1}$ whenever $\eta$ is as above and $\Psi$ is a full rank matrix. The conclusion to draw is that when $u_0 \in W = PW_{\lambda_k}$, we have $\text{MSE}(\Omega) = \sigma^2tr(C^{-1}_{(W,\Omega)})$ and $\eqref{trinv_problem}$ becomes 
\begin{equation}\label{pw_trinv_problem}
\begin{aligned}
\min_\Omega &\quad \text{MSE}(\Omega) = \sigma^2 tr(C^{-1}_{(W,\Omega)}) \\
\text{s.t.} &\quad |\Omega| = s.
\end{aligned}
\end{equation}

\begin{proposition}
 Suppose $W = PW_{\lambda_k}$ has dimension $m < d$, and let $C_\Omega$ and $C_{(W,\Omega)}$ be as above. Then $tr(C^{-1}_{(W,\Omega)} )< tr(C^{-1}_\Omega)$, and hence $\mathbb{E}(\|\epsilon_{(W,\Omega)}\|^2) < \mathbb{E}(\|\epsilon_\Omega\|^2)$
\end{proposition}

\begin{proof}
    Let $I_W: W \hookrightarrow \ell^2(V)$ be the natural inclusion as in Definition \ref{AOmatrix}. We can explicitly write
    \begin{equation}\label{frame_op_pw}
    \begin{split}
    C_{(W,\Omega)} &= \sum_{n = 0}^L I_W^*A^n S_\Omega A^n I_W \\
    &= I_W^*(\sum_{n = 0}^L A^n S_\Omega A^n) I_W \\
    &= I_W^* C_\Omega I_W.
    \end{split}
    \end{equation}
    Thus, if $0 < \rho_1 \leq...\leq \rho_m$ are the eigenvalues of $C_{(W,\Omega)} = I_W^*C_\Omega I_W$, and $0 < \mu_1 \leq...\leq \mu_d$ are the eigenvalues of $C_\Omega$, then the Poincare Separation Theorem (see \cite{MN19}) guarantees that $\rho_i \geq \mu_i$, hence $1/\rho_i \leq 1/\mu_i$. Therefore 
    \[
    tr(C_{(W,\Omega)}^{-1}) = \sum_{i = 1}^m \frac{1}{\rho_i} < \sum_{i = 1}^d \frac{1}{\mu_i} = tr(C_{\Omega}^{-1})
    \]
    as claimed.
\end{proof}

\smallskip
\noindent
We end this section with an explicit matrix form for $C_\Omega$ and $C_{(W,\Omega)}$. Let $A = BDB^*$ as before. Then $C_\Omega$ is computed as:
\begin{align*}
    C_\Omega &= \sum_{n = 0}^{L-1} A^n S_\Omega A^n \\
             &= \sum_{n = 0}^{L-1} BD^nB^*S_\Omega B D^n B^* \\
             &= U\bigg( \sum_{n=0}^{L-1} \sum_{v \in \Omega}D ^n B^*S_{\{v\}}B D^n\bigg)B^*.
\end{align*}
Writing $D = diag(\lambda_1,...,\lambda_d)$, we have that the $ij^{th}$ entry of $D^nU^*S_{\{v\}}UD^n$ is
\[
(\lambda_i\lambda_j)^nb_{vi}\overline{b_{vj}}
\]
Thus, the $ij^{th}$ entry of $B^*C_\Omega B$ is
\begin{equation}\label{BCB_entries}
\begin{cases}
\frac{1-(\lambda_i\lambda_j)^L}{1-\lambda_i\lambda_j}\sum_{v\in \Omega}b_{vi}\overline{b_{vj}},& \lambda_i\lambda_j \neq 1; \\
L\sum_{v\in \Omega} b_{vi}\overline{b_{vj}}, & \lambda_i\lambda_j = 1.
\end{cases}
\end{equation}
Note that $B^*C_\Omega B$ is the matrix form of $C_\Omega$ with respect to the basis of eigenvectors of $A$. If $W = PW_{\lambda_k}$ has dimension $m$, then (\ref{frame_op_pw}) shows that $C_{(W,\Omega)} = C_\Omega|_{W}$. Therefore, the matrix form of $C_{(W,\Omega)}$ with respect to the basis of eigenvectors spanning $W$ is the principal $m\times m$ submatrix of $B^*C_\Omega B$. 

\smallskip \noindent
In practice, we are often concerned with signals which decay in magnitude over time. This means the norm of the graph operator $A$ has norm less than 1. To model such evolution processes, we will assume that $\sigma(A) \subset (0,1)$ in all numerical experiments in the next section. Furthermore, it was shown in \cite{AHKLLV18} that, in such a scenario, MSE$(\Omega)$ converges for all subsets $\Omega \subset V$ as $L \to \infty$. For all numerical considerations, we will assume to take infinite samples in time, in which case the formula for $B^*C_\Omega B$ becomes: 
\begin{equation}\label{BCB_entries2}
\frac{1}{1-\lambda_i\lambda_j}\sum_{v\in \Omega}b_{vi}\overline{b_{vj}}
\end{equation}
\subsection{Numerical Algorithms} \label {NA}

The optimization problem (6) may not be feasible to solve using brute-force methods when the number of vertices in the graph is $20$ or more. In this section, we propose several numerical algorithms for obtaining approximate solutions to \eqref{pw_trinv_problem}. We recast the optimization problem $\eqref{pw_trinv_problem}$ as a binary optimization problem as follows.

\smallskip \noindent
Recall that $W = PW_{k}$ and note that 
\[
C_{(W,\Omega)} = \sum_{v \in \Omega} C_{(W,\{v\})}.
\]
Given a vector $x = (x_v)_{v\in V} \in \ell^1(V)$, we use the notation 
\begin{equation}\label{eqCW}
   C_W(x) = \sum_{v \in V} x_vC_{(W,\{v\})}. 
\end{equation}
With this notation, problem \eqref{pw_trinv_problem} becomes equivalent to 
\begin{equation}\label{bin_trinv_problem}
\begin{aligned}
\min_{x\in \ell^{1}(V)} &\quad tr(C_W(x)^{-1}) \\
\text{s.t.} &\quad \|x\|_1 = s, \quad x_v \in \{0,1\},\quad v\in V.
\end{aligned}
\end{equation}
In the following, we propose several numerical algorithms for solving the above problem \eqref{pw_trinv_problem}. We begin with a fast greedy worst-out best-in algorithm, whose performance depends heavily on its initialization. We then explore several efficient algorithms based on relaxing the binary constraints. It was first proposed in \cite{BJ09} to simply relax the binary constraints to convex constraints, and round the output to the nearest binary vector. That is, to solve the optimization problem

\begin{equation}\label{relaxed_trinv_problem}
\begin{aligned}
\min_{x\in \ell^{1}(V)} &\quad tr(C_W(x)^{-1}) \\
\text{s.t.} &\quad \|x\|_1 = s, \quad x_v \in [0,1],\quad v\in V.
\end{aligned}
\end{equation}

\smallskip \noindent
We will show heuristically that this method does not effectively minimize the objective in the case of dynamical sampling on graphs. For this reason, we will propose modifications to this method which penalize the objective function to promote a binary solution. Heuristically, we show that these algorithms produce a good approximation of the solution. Using this approximation as an initialization for the greedy algorithm typically produces the optimal solution.

\subsubsection{Greedy Algorithm} \label {GA}
We begin with the algorithm based on the greedy selection of the graph nodes. In the pseudo-code given below, we use $e_v \in \ell^1(V)$ to denote the $v$-th standard basis vector. Starting with a randomly selected binary vector $x_0$, which describes the initial selection of $s$ nodes, the algorithm computes which node will lead to the smallest increase of the objective function when it is thrown out from the selection. In the next step, the algorithms finds a node to add back to the selection which will decrease the objective function as much as possible. In case of multiple minima in either of the two steps the algorithm will pick the smallest numbered node. The out-in procedure is repeated as long as the objective function improves.

\begin{tabular}{rp{8cm}}
\toprule
\multicolumn{2}{p{13cm}}{\textbf{Algorithm 2.} Pseudo-code for Worst-Out Best-In Greedy Algorithm.}\\
\midrule
1:& \textbf{Input:} Initial binary vector $x_0$ with $\|x_0\|_1 = s$. \\
2:& $x=x_0$, $y = \textbf{0}$\\
3:&\textbf{while} $\|x-y\|_1 > 0$\\
4:& \quad $y = x$\\
5:& \quad \textbf{set} \[v_\text{out} := \argmin_{\{v \in V \mid x_v = 1\}}\; \text{tr}\;C_{W}(x - e_v)^{-1}\] \\
6:& \quad $x_{v_\text{out}} = 0$ \\
7:& \quad \textbf{set} \[v_\text{in} := \argmin_{\{v \in V \mid x_v = 0\}}\; \text{tr}\;C_{W}(x + e_v)^{-1}\] \\
8:& \quad $x_{v_{in}} = 1$ \\
9:& \textbf{Output:} New binary vector $x$ with $\|x\|_1 = s$.\\ 
\bottomrule
\end{tabular}

\smallskip 
\noindent
Unsurprisingly, the accuracy of the approximation to the solution of \eqref{bin_trinv_problem} provided by the greedy method depends heavily on the initialization procedure. In what follows, we present algorithms that provide an initialization that significantly outperforms random selection. In fact, in many cases, the vector obtained from this initialization will itself be a solution of \eqref{bin_trinv_problem}. 

\subsubsection{Norm Penalty} 
The algorithm of this section is based on the classical idea of relaxing the discrete feasible set of the binary program \eqref{bin_trinv_problem} to its convex hull and adding a penalty term that encourages binary solutions.

\smallskip 
\noindent
Let $\lambda \ge 0$ be a parameter chosen by the user.
The problem \eqref{bin_trinv_problem}
is equivalent to
\begin{equation}\label{add_pen_problem}
\begin{aligned}
\min_{x \in \ell^1(V)} &\quad {\mathrm{tr} } (C_W(x)^{-1}) +\lambda(\|x\|_1 - \|x\|_2^2) \\
\text{s.t.} &\quad \|x\|_1 = s, \;\; x_v \in \{0, 1\}, \;\; v\in V(G),
\end{aligned}
\end{equation}
since $\|x\|_1 =  \|x\|_2^2$ for binary vectors $x \in \ell^1(V)$. Relaxing the binary constraints yields
\begin{equation}\label{DC_problem}
\begin{aligned}
\min_{x \in \ell^1(V)} &\quad G_\lambda(x) = {\mathrm{tr} } (C_W(x)^{-1}) +\lambda(\|x\|_1 - \|x\|_2^2) \\
\text{s.t.} &\quad \|x\|_1 = s, \;\; x_v \in [0, 1], \;\; v\in V(G),
\end{aligned}
\end{equation}
which is an example of a DC-program (``DC'' stands for a difference of convex). DC-programming is a well-developed area of non-convex optimization \cite{LtPD18}. Such programs are usually solved via a DC algorithm (DCA).
Given that the function $G_\lambda$ is infinitely differentiable (on its domain), we chose to use the concave-convex procedure (CCCP \cite{YR03}) as the version of DCA we apply. Specifically, this iterative process updates as:
\[
x_{n+1} = \argmin [{\mathrm{tr} } (C_W(x)^{-1}) + 
\lambda(s-2\langle x, x_n\rangle)] = \argmin [{\mathrm{tr} } (C_W(x)^{-1}) - 2\lambda\langle x, x_n\rangle].
\]
Note that this process will, generally, not converge to a binary solution. For this reason, we round the output to the closest binary vector which satisfies the constraints.
\begin{remark}\label{SY-remark}
    We could also use a multiplicative penalty, which does not depend on a choice of parameter:
    \begin{equation}
\begin{aligned}
\min_{x \in \ell^1(V)} &\quad G_1(x) = {\mathrm{tr} } (C_W(x)^{-1}) \cdot\frac{\|x\|_1}{\|x\|_2^2} \\
\emph{s.t.} &\quad \|x\|_1 = s, \;\; x_v \in [0, 1], \;\; v\in V(G).
\end{aligned}
\end{equation}
Taking the reciprocal of the objective, we have the equivalent optimization problem:
\begin{equation}\label{SY_alg}
\begin{aligned}
\max_{x \in \ell^1(V)} &\quad G_2(x) = \frac{\|x\|_2^2}{ \|x\|_1 \cdot{\mathrm{tr} } (C_W(x)^{-1})} \\
\emph{s.t.} &\quad \|x\|_1 = s, \;\; x_v \in [0, 1], \;\; v\in V(G).
\end{aligned}
\end{equation}
This is an example of a convex-convex quadratic fractional program. Using the methods discussed in \emph{\cite{KM23}}, we can suboptimally solve \eqref{SY_alg} via the iterative process:
\[
x_{n+1}  = \argmin [{\mathrm{tr} } (C_W(x)^{-1}) - 2\lambda_n\langle x, x_n\rangle],
\]
where $\lambda_n = \frac{F(x_n)}{\|x_n\|_2^2}$, $F(x) = {\mathrm{tr} } (C(x)^{-1}).$ We note that the implementation of this algorithm is essentially the same as the method used to solve the DC-program \eqref{DC_problem}.
\end{remark}

\begin{remark}
    Another interesting penalty uses the same objective function proposed in Remark \ref{SY-remark}, but constrains the 2-norm of the input vector instead of the 1-norm. More precisely: 
    \begin{equation}\label{iki_problem}
\begin{aligned}
\min_{x \in \ell^2(V)} &\quad H_2(x) = {\mathrm{tr} } (C_W(x)^{-1}) \cdot\frac{\|x\|_1}{\|x\|_2^2} \\
\text{s.t.} &\quad \|x\|_2^2 = s, \;\; x_v \in [0, 1], \;\; v\in V(G).
\end{aligned}
\end{equation}
To obtain a feasible solution of \eqref{bin_trinv_problem}, one may simply round the $s$ largest coordinates of the output to 1 and set the rest to 0. The output of this algorithm is usually sparse, and if it has less than $s$ nonzero coordinates, one may use greedy selection \emph{\cite{CPR17}} to select for the remaining vertices. We have observed this to be an efficient algorithm when the objective is minimized using the standard SciPy optimizer. However, this objective function is nonconvex, so it is unclear why this method is sensible. For this reason, we have chosen to omit it from the numerical experiments.
\end{remark}
\subsubsection{Exponential Penalty}
Another penalty method we propose is the use of exponential penalty terms. Specifically, we choose a small number $\delta > 0$ and solve 
\begin{equation}\label{EP_problem}
\begin{aligned}
\min_{x \in \ell^1(V)} &\quad G_\delta(x) = {\mathrm{tr} } (C_W(x)^{-1}) +\sum_{v \in V(G)} e^{-x_v/\delta} \\
\text{s.t.} &\quad \|x\|_1 = s, \;\; x_v \in [0, 1], \;\; v\in V(G).
\end{aligned}
\end{equation}
This method has the advantage of being a single convex programming problem instead of an iterative procedure. Typically, this method does not produce a sparse solution. Nevertheless, the rounded solution almost always outperforms the unpenalized convex relaxation.

\subsection {Numerical experiments} \label {NP}
In this section we conduct a series of numerical tests which demonstrate the efficiency of the algorithms proposed in the previous sections. In the first test, we demonstrate that the performance of the algorithms, relative to one another, is stable under different distributions of eigenvalues of the graph operator $A$, and different types of graphs. This will justify the choice to use uniform random eigenvalues and random connected graphs for the remainder of the experiments. The next tests demonstrate the accuracy of each algorithm, comparing them to the convex relaxation proposed in \cite{BJ09} and greedily minimizing the trace inverse \cite{CPR17}. We also test the accuracy of the Worst-Out Best-In (Wo-Bi) greedy algorithm initialized with the output of each algorithm. 

\subsubsection{Consistency across various graph types and eigenvalue distributions}
To show the consistency of the performance of the proposed algorithms under the type of graph and distribution of eigenvalues, we conduct a numerical test as follows. We choose four probability distributions: uniform distribution $\text{U}(0,1)$, triangular distribution $\text{Triangular}(0,\frac{1}{2},1)$, beta distribution $\text{Beta}(\frac{1}{2},\frac{1}{2})$, and exponential distribution $\text{Exp}(1)$. In each test, the eigenvalues of the graph operator will be drawn according to these distributions. The eigenvectors of the graph operator will correspond to the eigenvectors of the Laplacian of a random graph, a clustered graph, a path graph, and a graph with a circulant adjacency matrix. All graphs in this test have 45 nodes, the number of nodes to select is fixed to $12$, and the dimension of the Paley-Wiener space is fixed to $10$. By a clustered graph, we mean two complete graphs, one of size 38 and another of size 7, each with two random edges removed, connected by two edges (see Figure \ref{fig:clustered}). 

\begin{figure}[H]
    \centering
    \includegraphics[width=0.5\linewidth]{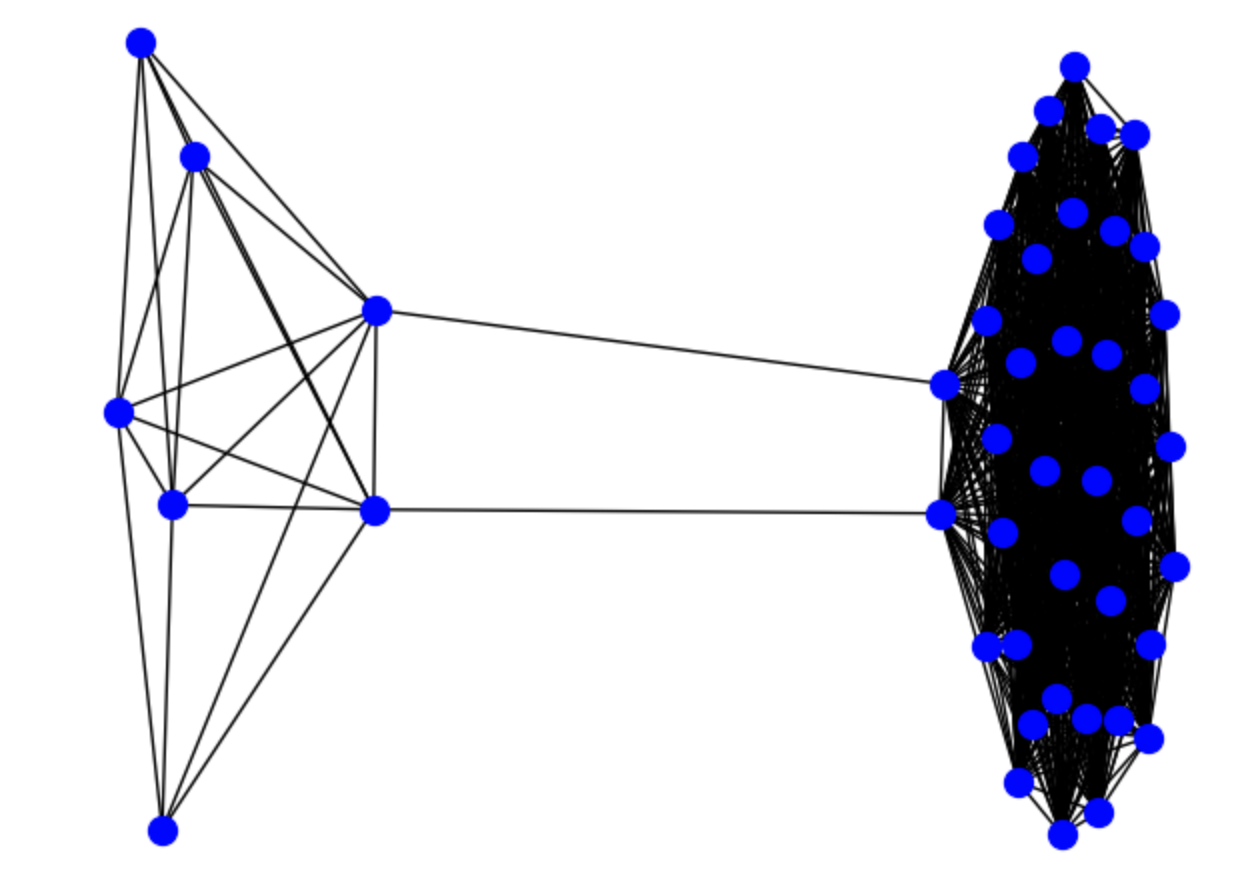}
    \caption{Clustered Graph}
    \label{fig:clustered}
\end{figure}

\begin{figure}[H]   
    \centering
    {\includegraphics[width=0.45\linewidth]{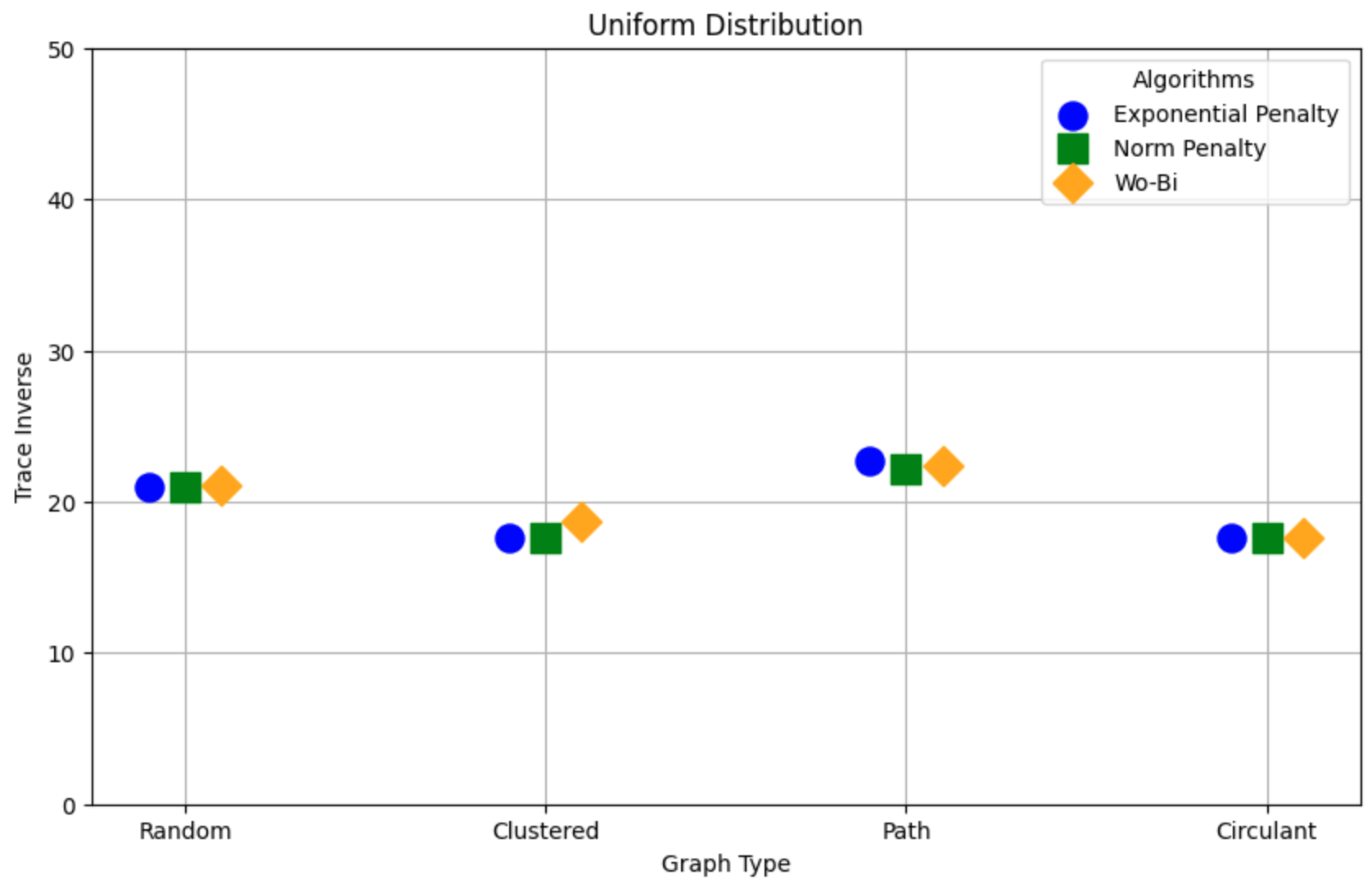}}
    \qquad
    {\includegraphics[width=0.45\linewidth]{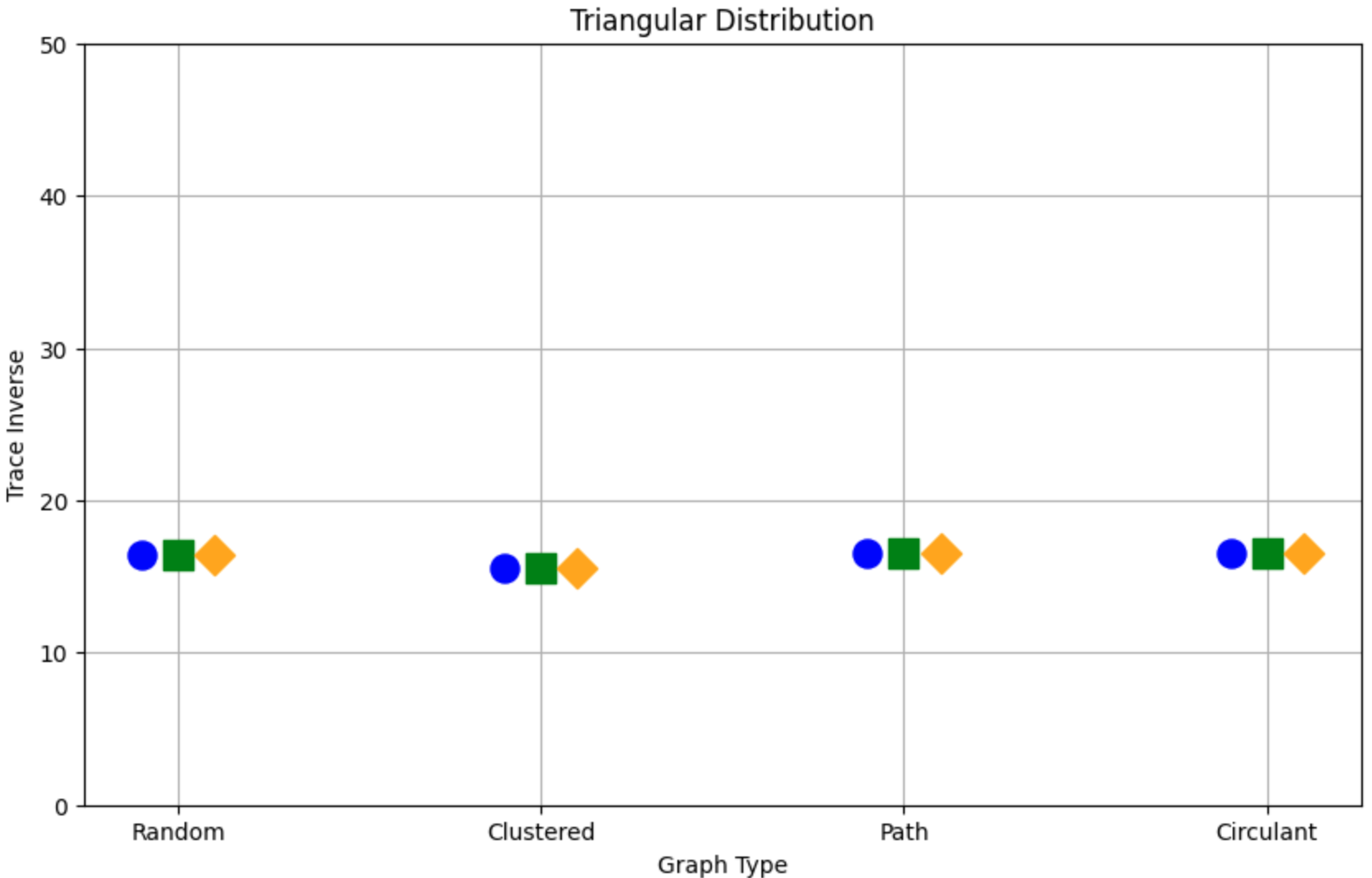}}

    {\includegraphics[width=0.45\linewidth]{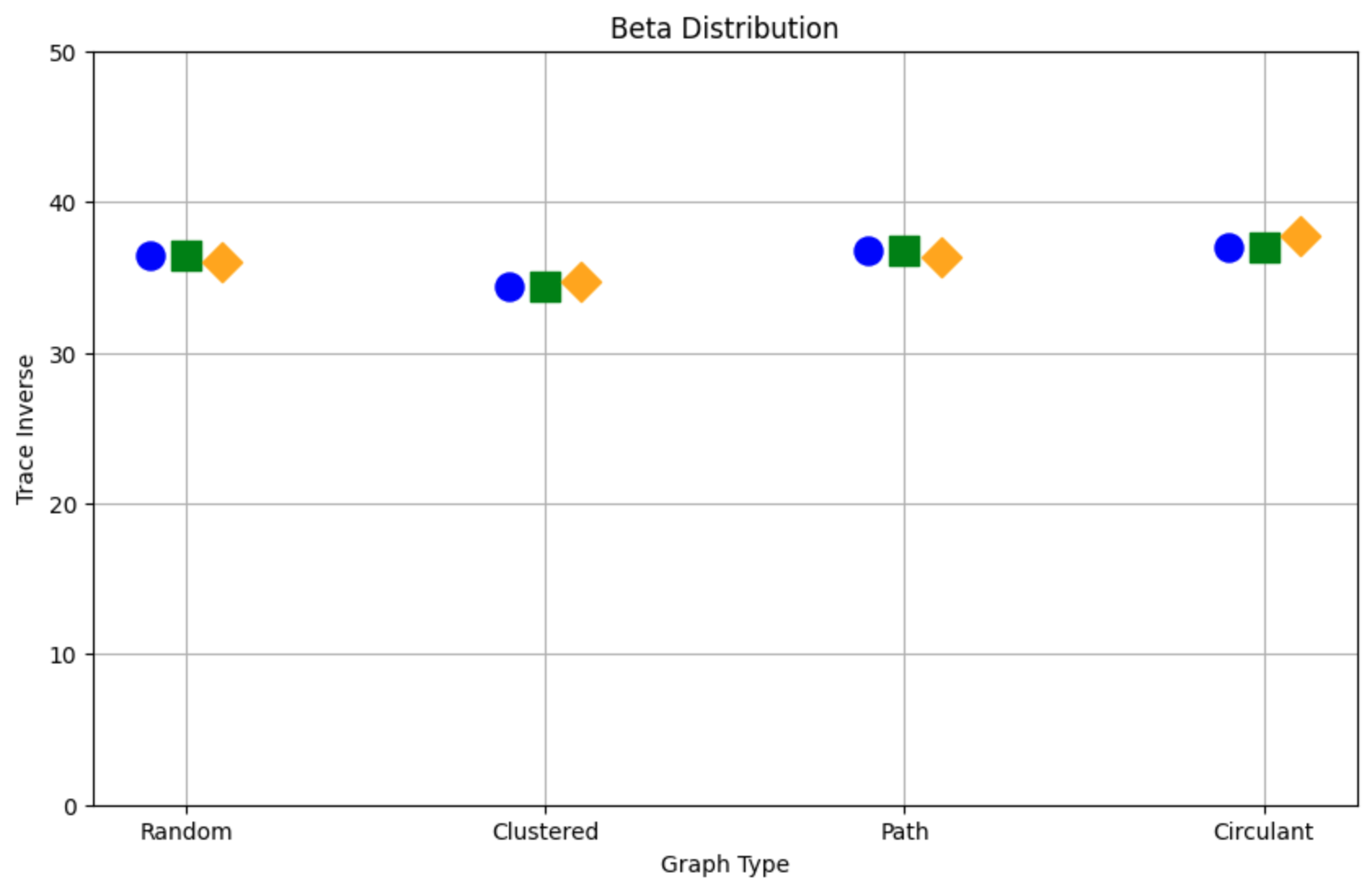}}
    \qquad
    {\includegraphics[width=0.45\linewidth]{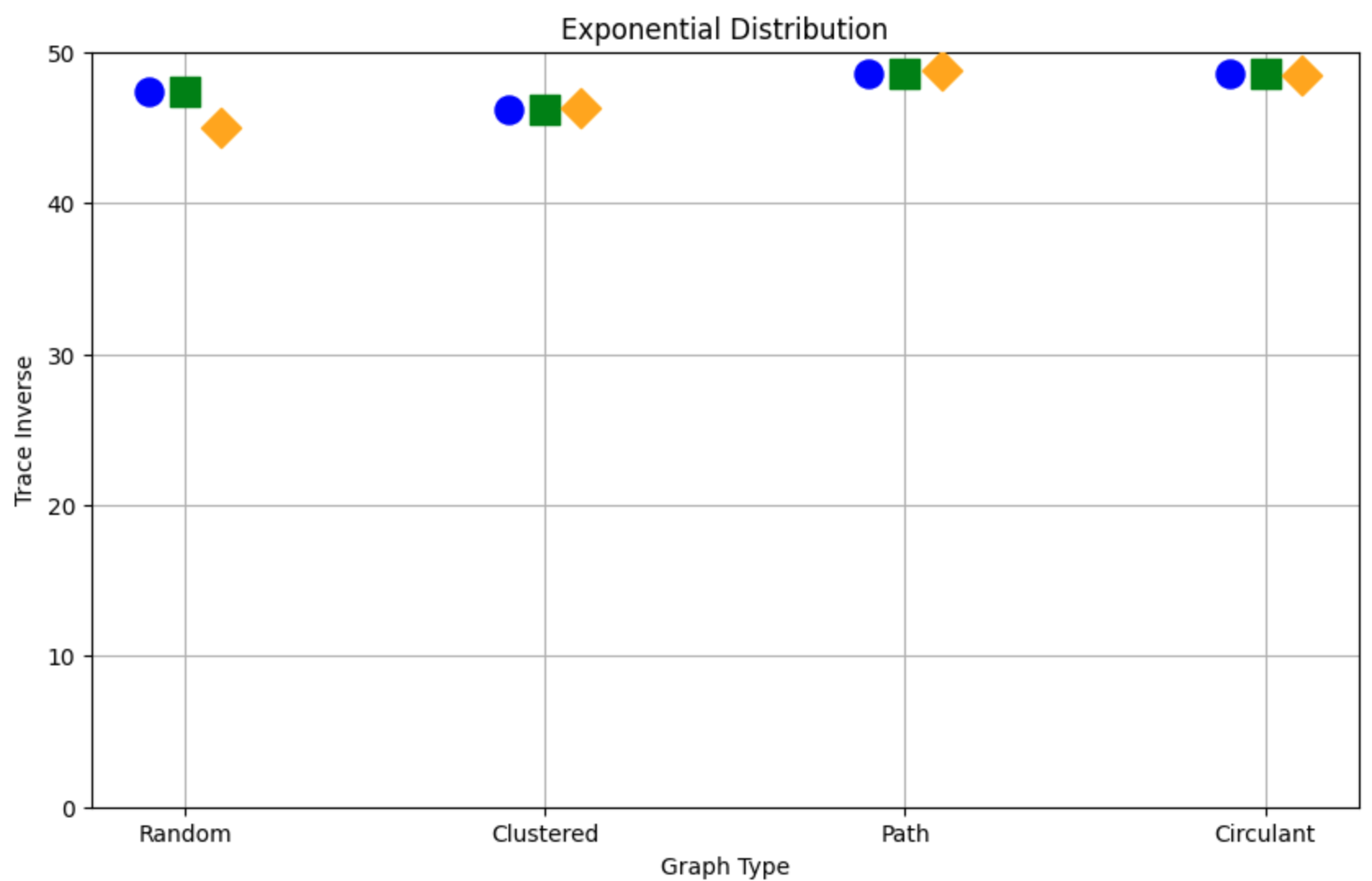}}
    \caption{Four figures showing the output of each algorithm varying the type of graph and distribution of eigenvalues of graph operator. }
    \label{fig:distribution_test}
\end{figure}

\smallskip
\noindent
Figure \ref{fig:distribution_test} above shows the output of the exponential penalty algorithm, norm penalty algorithm, and the Wo-Bi greedy algorithm. As the figure demonstrates, each of the algorithms is performing consistently relative to the others despite changing the distribution of the eigenvalues of the graph operator and the type of graph. For this reason, we will use only random connected graphs and uniform random eigenvalues for the graph operator in all subsequent tests. 

\subsubsection{Comparison to the brute force algorithm for small graphs}

Below we demonstrate the accuracy of the proposed algorithms against the brute-force algorithm (i.e. one that simply checks all feasible binary points of \eqref{bin_trinv_problem} and selects the optimal). For this test, we again use random graphs of sizes 15 and 20, and compare the accuracy of each algorithm and the time taken to run. 

\begin{figure}[H]
    \centering
    {\includegraphics[width=0.45\linewidth]{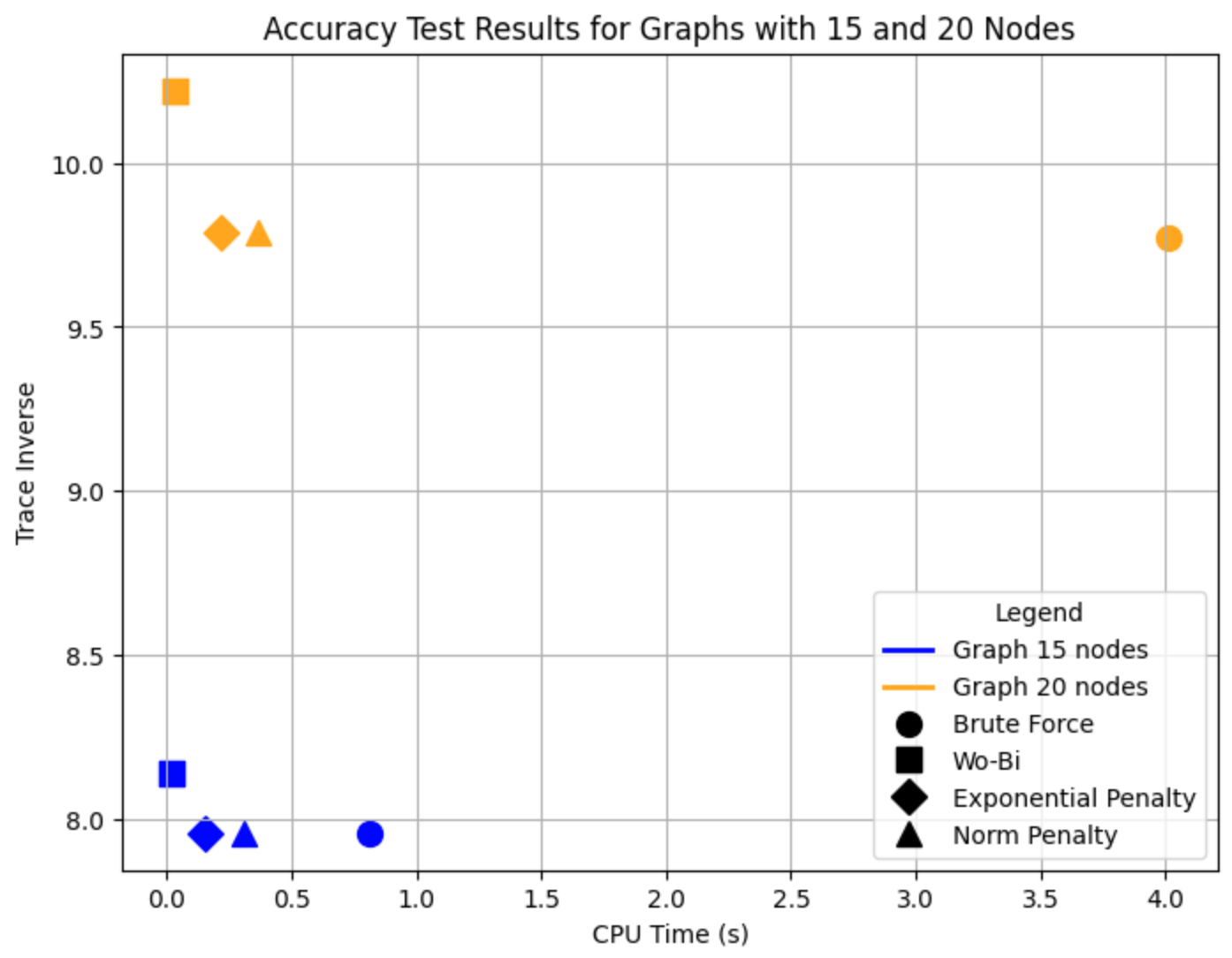}}
    \qquad
    {\includegraphics[width=0.45\linewidth]{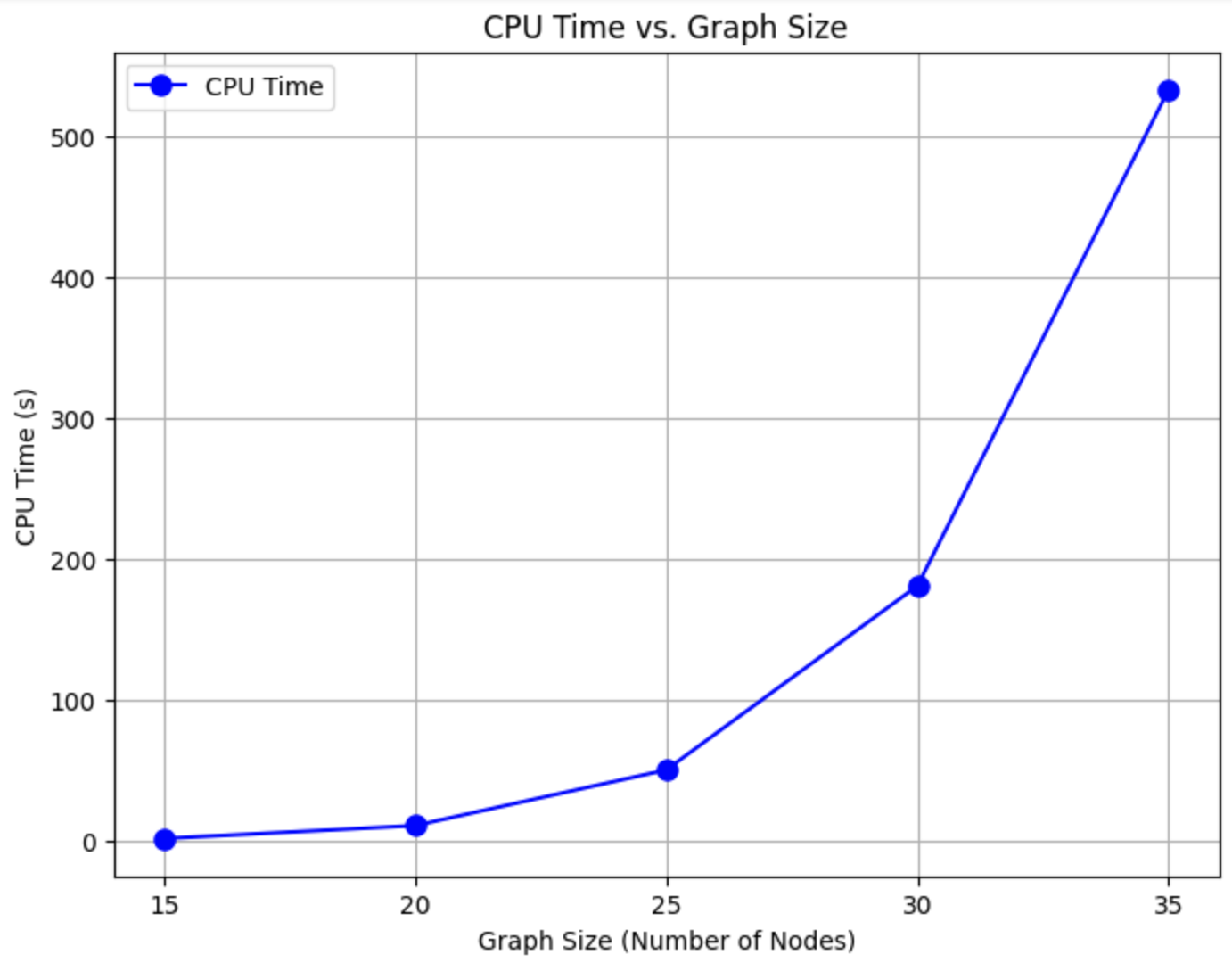}}
    \caption{Median Trace Inverse and CPU time of each algorithm (left) and CPU time for brute force algorithm (right). }
    \label{brute_force}
\end{figure}

\noindent For the figure above, we use random graphs with 15 and 20 nodes each having 45 edges and 50 edges respectively. The dimension of the Paley-Wiener space is taken to be 5, and the number of nodes to select is fixed to be 6. We notice that the exponential penalty and norm penalty perform nearly identically to the brute force algorithm in terms of accuracy, and are far faster. The Wo-Bi algorithm with random initialization performs slightly worse, but is notably faster than all the others. On the right, we see how the brute force algorithm quickly becomes infeasible to perform even for graphs of modest sizes. For example, this algorithm takes over 3 minutes to perform on a graph of size 30. In the next section we will see that all other algorithms take less than one second for graphs of this size.

\subsubsection{Relative accuracy comparison between algorithms}
To test the accuracy of each algorithm, we choose graph operators having uniform random eigenvalues and eigenvectors corresponding to Laplacians of random connected graphs with vertex sets of sizes 15, 30, 45, 60, 75, and 100, and the size of a subset of nodes to be respectively 6, 8, and 12 for larger graph sizes. The dimension of the Paley-Wiener space is always fixed to be 2 less than the subset size (i.e.~4, 6, and 10). All data points shown are averages of 100 tests. For some of the algorithms, some data points are missing. This is because we omit any data points corresponding to an average trace inverse larger than 100. 

\begin{figure}[H]
    \centering
    {\includegraphics[width=0.45\linewidth]{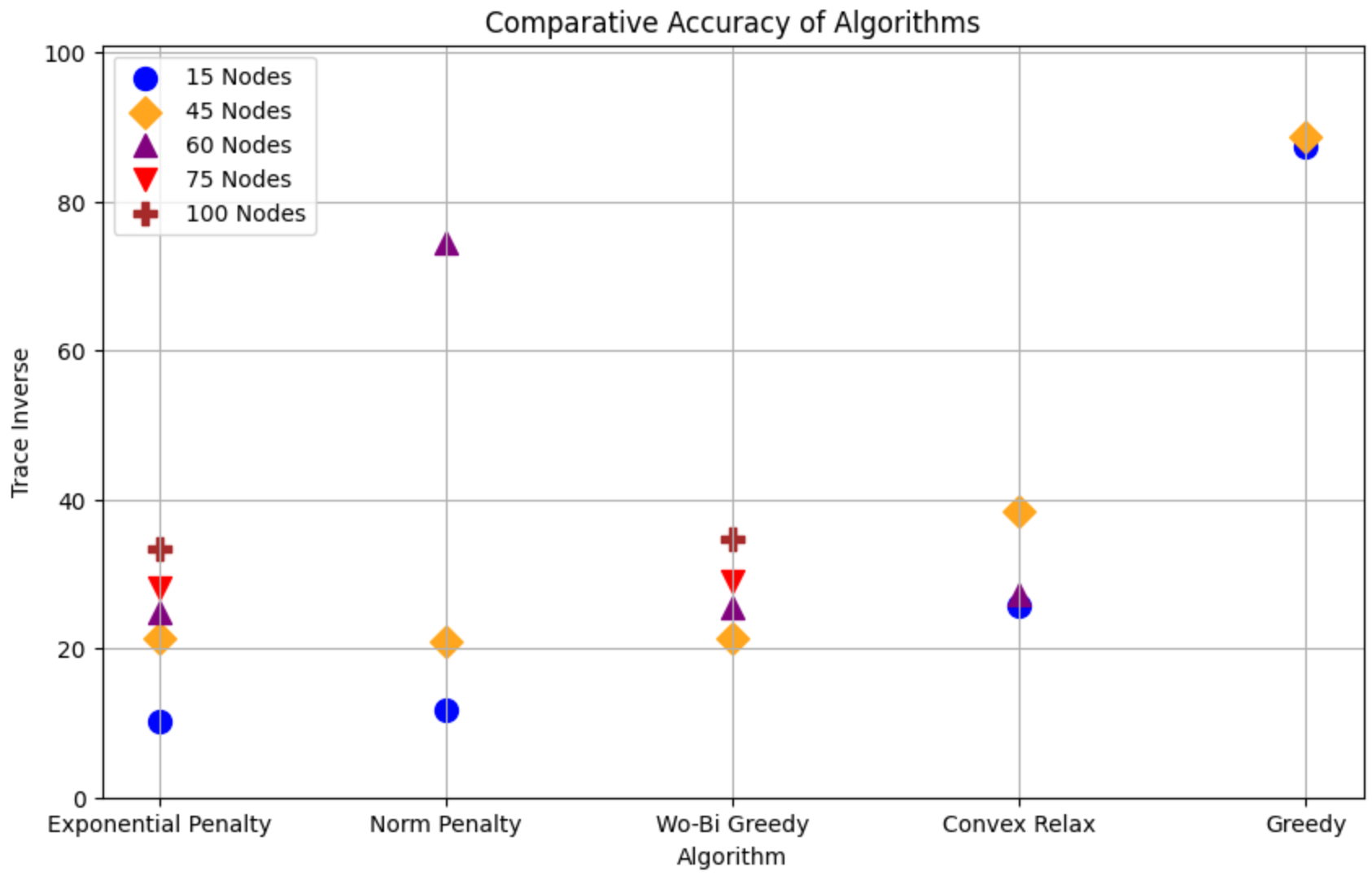}}
    \qquad
    {\includegraphics[width=0.45\linewidth]{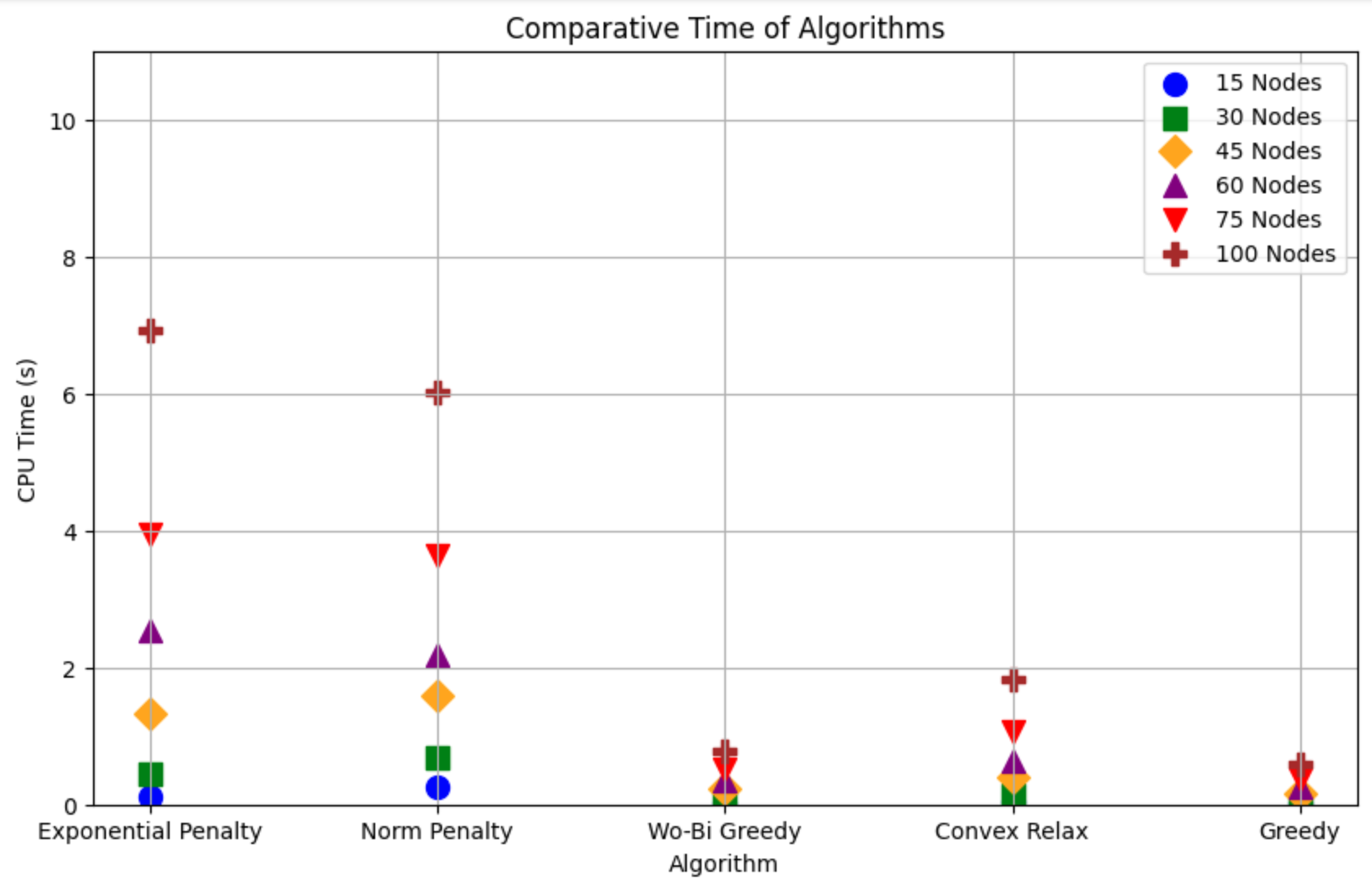}}
    \caption{Average trace inverse of the output of each algorithm (left) and the comparative time taken (right) run on random graphs. }
    \label{time_test}
\end{figure}

A few comments are in order here. We note that the unpenalized convex relaxation \eqref{relaxed_trinv_problem} and the standard greedy algorithm do not perform well in these tests. Both algorithms perform very poorly for even modestly sized graphs. The exponential penalty and the Wo-Bi algorithm consistently perform the best and are usually comparable to one another. Furthermore, the norm penalty algorithm performs well for small graphs, but not for large graphs. Moreover, the algorithms using convex optimization clearly take the longest time to run, with the penalized versions \eqref{EP_problem} and \eqref{DC_problem} taking significantly longer. Finally, it should be noted that for $|V| = 15$, the average trace inverse value for the Wo-Bi algorithm was over 100. This algorithm relies heavily on the initialization, and although a random initialization usually produces a good output, it can occasionally perform very poorly.

\smallskip \noindent
We have found that the outputs of each algorithm make for good initializations for the Wo-Bi greedy algorithm. Below we demonstrate this by running each algorithm as above, then using the output of each algorithm as an initialization for the Wo-Bi algorithm and recording the trace inverse. 

\begin{figure}[H]
    \centering
    \includegraphics[width=0.5\linewidth]{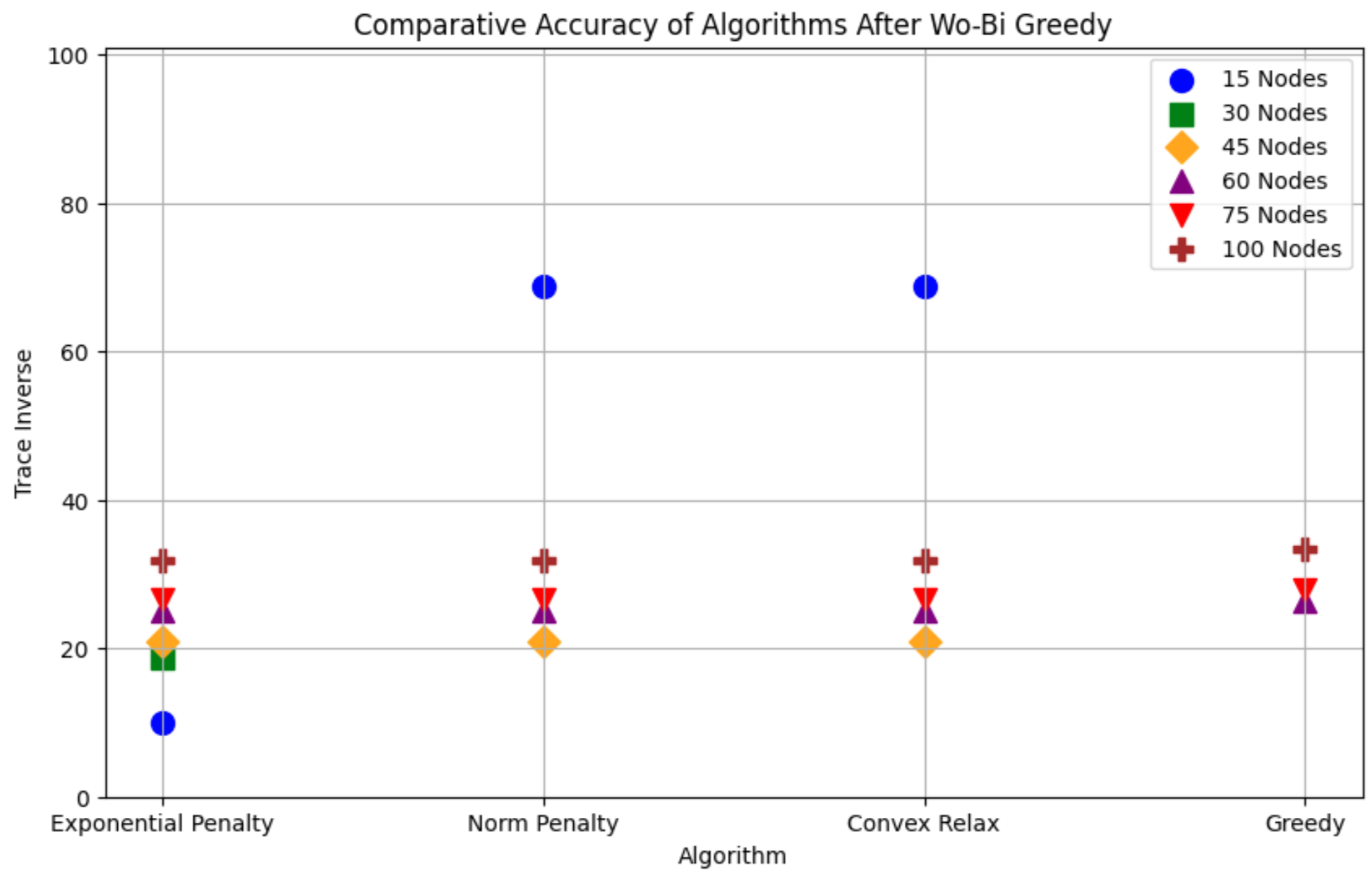}
    \caption{Average trace inverse of each algorithm when combined with the Wo-Bi algorithm.}
    \label{fig:after_arm}
\end{figure}

\smallskip\noindent
To conclude, we see that the exponential penalty combined with the Wo-Bi algorithm yield a combined algorithm more accurate than those previously proposed. However, the exponential penalty is the most computationally expensive, and would not be appropriate for very large graphs. On the other hand, Figure \ref{fig:after_arm} shows that for graphs of large sizes, the other algorithms perform similarly to the exponential penalty when combined with the Wo-Bi algorithm. Thus, it is most appropriate to use the exponential penalty combined with the Wo-Bi algorithm for small graphs, and a faster algorithm for Wo-Bi initialization on large graphs.


\bibliographystyle{siam}
\bibliography{shortref}

\end{document}